%% file: BlackWhiteArray.tex
\documentclass[11pt]{article}  % also works

\usepackage{blindtext}
\usepackage{geometry}                % See geometry.pdf to learn the layout options. There are lots.
\usepackage{graphicx}
\usepackage{amssymb}
\usepackage{epstopdf}

\usepackage{multicol}

\usepackage{environ}

\usepackage{listings}

\usepackage[utf8]{inputenc}

\usepackage{listings}
\usepackage{color}

\usepackage{tikz}
\usetikzlibrary{shapes,arrows}

\usepackage{tikz-3dplot}
\usetikzlibrary{calc,3d}

{%\kaiti
\addtolength\leftmargini{0.05in}
\addtolength\leftmargini{0.05in}

\begin{quote}
}
{
\end{quote}
}

{%\kaiti
\addtolength\leftmargini{0.05in}
\addtolength\leftmargini{0.05in}
\begin{lstlisting}[mathescape=true]
}
{
\end{lstlisting}
}

{%\kaiti
\addtolength\leftmargini{0.05in}
\addtolength\leftmargini{0.05in}
\begin{lstlisting}[mathescape=true]
}
{
\end{lstlisting}
}

\NewEnviron{example}{%
\begin{equation}\begin{split}
  \BODY
\end{split}\end{equation}
}

\NewEnviron{example2}{%
\begin{lstlisting}[mathescape=true]
  \BODY
\end{lstlisting}
}

{\begin{quote}\begin{tiny}\it}%
{\end{tiny}\end{quote}}

{\begin{lstlisting}
\begin{tiny}\it}%
{\end{tiny}
\end{lstlisting}}

\lstdefinestyle{mystyle}{
    commentstyle=\color{codegreen},
    stringstyle=\color{codepurple},
    basicstyle=\footnotesize,
    breakatwhitespace=false,         
    breaklines=true,                 
    captionpos=b,                    
    keepspaces=true,                 
    numbers=left,                    
    numbersep=5pt,                  
    showspaces=false,                
    showstringspaces=false,
    showtabs=false,                  
    tabsize=2,
    framexleftmargin=100mm,
    basicstyle=\ttfamily,
    xleftmargin=1cm
}

\usepackage[english]{babel}
 
\makeatletter
\def\CTEX@section@format{\Large\bfseries}
\makeatother

\makeatletter
\g@addto@macro{\CTEX@section@format}{\raggedright}
\makeatother

\usepackage{geometry}
\usepackage{titlesec}
\titleformat{\section}
 {\normalfont\scshape}{\thesection}{0.5em}{}
 
 \titleformat{\subsection}
 {\normalfont\scshape}{\thesubsection}{0.3em}{}
\usepackage{float}
\newfloat{lstfloat}{htbp}{lop}
\floatname{lstfloat}{Listing}
\usepackage{csquotes}

\usepackage{listings} %For code in appendix
\usepackage{amsthm}
\usepackage{amsmath}

\newtheorem{lemma}{Lemma}
\newtheorem{theorem}{Theorem}

\theoremstyle{property}
\newtheorem{property}{Property}

\newtheorem{corollary}{Corollary}

\title{Black-White Array: A New Data Structure for\\
  Dynamic Data Sets }
\author{Z. George Mou\\
                {\footnotesize gmou5813@gmail.com}
              }
\date{}                                           % Activate to display a given date or no date

\begin{document}
\maketitle

%\scriptsize
%\tableofcontents

\abstract{

A new array based  data structure named black-white array (BWA) is introduced as an effective and efficient alternative to the
list or tree based data structures for dynamic data set. It consists of two sub-arrays, one white and one
black of half of the size of the white. Both of them are conceptually partitioned into segments of different ranks with the sizes grow in geometric sequence. The  layout of BWA allows easy calculation of the meta-data  about the segments, which are used extensively in the algorithms for the  basic operations of the dynamic sets.

The insertion of a sequence of unordered numbers into BWA takes amortized time logarithmic to the length of the sequence. It is also proven
that when the searched or deleted value is present in the BWA, the  asymptotic amortized cost for the operations  is  $O(\log (n))$; otherwise, the 
time will  fall somewhere between $O(\log(n))$ and $O(\log^2(n))$.

It is shown that the state variable {\tt total}, which records the number of values in the BWA captures the dynamics of state transition of BWA. 
This fact is exploited to produce concise, easy-to-understand, and efficient coding for the operations. As it uses arrays as the 
underlying structure  for  dynamic set,  a BWA need neither the space to store the pointers referencing other data nodes nor the time to chase the pointers as with any linked data structures. 

A C++ implementation of the BWA  is completed. The performance data were gathered and plotted, which confirmed the theoretic analysis. The testing results showed 
 that the amortized time for the insert, search, and delete  operations is all  just between 105.949 and 5720.49 nanoseconds for  BWAs of sizes ranging 
from $2^{10}$ to $2^{29}$ under various conditions.}


%\tableofcontents

%\begin{multicols}{2}

\section{Introduction}

The size of a dynamic data set (DDS) may grow and shrink over time with  insert, search, and delete operations \cite{corn90}. The values in DDS are assumed to be drawn from some linearly ordered domains\footnote{More formally, totally ordered sets} such as integers, real numbers, lexicographically ordered words, and keys as in systems such as database.
The linearity of the domain, though may not stated explicitly, is of critical importance  to the implementations of the dynamic data sets.

It can be observed that linked structures, such as  lists and trees, are almost exclusively used in the implementations of dynamic data sets \cite{aho74,sedgewick,corn90,redis,skiplist}. The  performance of the operations in these data structures is generally  tied to the heights of the trees. This  leads  to  the invention of a fairly large family of self-balancing data structures, including  the AVL tree by Georgy and Evgenii (1962) \cite{avl}, B-tree by Bayer and McCreight \cite{btree}, 2-3 Tree by Hopcroft  (1970) \cite{aho74}, Red-black tree by
Guibas and  Sedgewick (1978) \cite{rbtree} and many of their variants \cite{corn90}.

These tree based self-balancing data structures for dynamic data sets  all can perform the basic operations in time  logarithmic to  the number of the values in the structure, and in this sense they can be said to be optimal\footnote{The height of a balanced  tree with $n$ nodes has  necessarily a  height of $\Omega(n)$.}. However,  the flexibility of the tree structures does come at a price. First of all, additional space is needed for each data node  to store the pointers to the children,  parent (except the root), and some metadata related to the tree structure; 
secondly, extra time is required to perform the referencing and dereferencing of the data nodes since the values at the nodes can only be accessed via 
their  parents or children (pointer chasing); and finally, even though an elaborate scheme for tree rebalancing  may not affect the  asymptotic complexity, it could have yield some rather large constant  coefficient  in front of the logarithmic function.

This brings up the question: is it possible to use structures other than lists or trees, e.g. arrays, to implement a dynamic data  set with performance comparable to, or even in some aspects better than, those solutions based on linked structures? 

The work reported here is, in part, an intellectual inquiry into the question, and, in part, a exploration aimed at  providing a fast solution
to the real-world problem of point cloud registration \cite{pointcloud}.  The Iterative Closest  Point  (ICP) method \cite{icp} for  the problem calls for repeated neighborhood search in a three  dimensional space. For better efficiency, main stream method with  $k$-d trees \cite{kdtree}  for 3-d space search \cite{pcl} was not adopted.  Instead, a dimensional shuffle transform (DST) was  proposed \cite{mou-10pb, mou-13pb,mou-19pa,mou-18pct, mou-201b}, that reduces the  search of a region  in a high dimensional space \cite{sedgewick,samet} to a small number of regional search in one  dimensional space. This has been proven to lead to   order(s) of magnitude  improvement in searching speed by bypassing the high cost  normally associated with the maintenance, traversal, and rebalancing of  tree structures, which, in this case, is the $k$-d tree. The regional search over one dimensional space calls for DDS data structures. The idea of black-white array  was then conceived and experimented with \cite{mou-202}, which has  proven to be superior in performance to the list and tree based solutions  that were examined at the time\footnote{In particular, the skip list in Redis \cite{redis,skiplist} and red-black tress \cite{rbtree}.}.

 The paper is organized as follows. In Section \ref{sec:layout}, the layout of the black-white array is introduced; the algorithms in pseudocode for the operations are provided in Section \ref{sec:operations}; the asymptotic complexity of the operations are 
formally analyzed in Section \ref{subsec:analysis};  a C++ implementation and its  performance data are presented in Section \ref{subsec:benchmark}. Some extensions to the black-white array with augmented operations are discussed in Section \ref{sec:augmentations}; finally, concluding remarks are given in Section
\ref{sec:conclusion}.

\section{Layout}
\label{sec:layout}

A {\it black-white array} (BWA) of size $N = 2^k$ consists of a pair of arrays, the {\it black}  (B) and the {\it white} (W), with lengths of $N/2$ and $N$ respectively. 
The black array's indices ranges from 1 to $\frac{1}{2} N-1$, while the white from $1$ to $N-1$.

Both the black and white arrays are conceptually divided into segments of different ranks, where the segment  of rank $i$ contains the entries with indices in the closed interval of $[2^i, 2^{i+1} -1]$ for $i = 0$ to $k-1$. It follows that a segment of rank $i$ contains exactly $2^i$ entries,  and there are $k-1$ and $k$  segments in the black and white arrays respectively. An illustration of the BWA layout is provided in Figure \ref{fig:layout}.

Indexing of the $i$th entry of the black and white arrays will be written as $B[i]$ and $W[i]$ respectively. 
The starting and terminating indices of a segment of rank $j$ are respectively denoted by $S(j)$ and $T(j)$. It is easy to see that $S(j) = 2^j$, and $T(j) = 2^{j+1}-1$, which can be calculated with $S(j) = 1 \ll $ j, and $T(j) = (1 \ll (j+1)) -1$ respectively, where \enquote{$\ll $}  is the left shift of a (binary form of the) integer to the left by  the number of positions shown on the right side.

The black array is used only during the transitional states of the BWA, and at any stable state of the BWA, all the meaningful values are held in some segments of the white array. We say that a segment is {\it active} if the white segment with that rank is holding meaningful values. As will be seen,  it is important to decide  if a segment of a given rank is active at a given time.

A state variable  named {\tt total} is used to keep track  the number of values stored in the BWA\footnote{VOID produced by a delete operation is counted as a value, though special.}.  It turns  out that this variable captures the configuration of the segments in the sense that  a segment of rank $ i $ is active if and only if the $i$th least significant bit 
of {\tt total} has a value of one. 
We therefore can decide if segment of rank  $i $ is active 
by a simple micro: {\tt active} (i) = {\tt total}   \& (1 $\ll$ i), where \enquote{\&} is the bitwise AND,  which returns a non-zero value if and only if the segment of rank $i$ is active.

\begin{figure}[h]
\begin{center}
\resizebox{0.87\textwidth}{!} {      
\input{BWA_layout.tex}
}
\end{center}
\vspace{-0.2in}
\caption{The layout of a white-black array with four segments}
\label{fig:layout}
\end{figure}

\vspace{-3mm}
\section{Operations}
\label{sec:operations}

This section provides the procedures for the Insert, Search, and Delete operations of the BWA.
\vspace{-1mm}
\subsection{Insert}
\label{sec:insert}

%See Figure \ref{fig:insert}, see Listing \ref{prog_insert}

The Insert of a value  (Listing \ref{listing_insert}) is performed by first checking if the  rank 0 segment is active. If not, 
the  value is simply put at W[1] (Line 3); otherwise, it is put at B[1], followed by a merge of the black and white segments of rank 0 (Line 5, 6).
Note that  the \enquote{!} in Line 2  denotes the logical NOT operator.

\nopagebreak
%\noindent\rule{\linewidth}{0.4pt}
\begin{lstlisting}[caption={The Insert operation  of BWA. },label=listing_insert, captionpos=b]
Insert (v)
    if (!active(0))
        W[1] = v
    else
        B[1] = v
        merge (0)
    end
end
\end{lstlisting}   
%\noindent\rule{\linewidth}{0.4pt}
\pagebreak[2]

The following code is for the merge of black and white segments of  rank $i$. Again, the result of the merge is put in the white segment of next
higher rank, if it is not active; otherwise, the result is put in the black segment of higher rank, followed by a recursive call on the merge of two segments of that higher rank:
\begin{samepage}
\begin{lstlisting}[caption={The merge of  black and white segments of rank  i.},label=prog_merge, captionpos=b,abovecaptionskip=0.7em]
merge (i)
    s = S(i)
    t =  T(i)
    if (!active (i+1))
        merge (s, t, W))
    else
        merge (s, t, B)
        merge (i+1) 
    end
    total = total +1
end
\end{lstlisting}   
\end{samepage}

	Note that the operator {\tt merge} is overloaded:  {\tt merge (i) } refers to the merge of  black and white segments of rank $i$, whereas {\tt merge (s, t, arr) } refers the merge of sections of black and white arrays with indices in the range of [s, t]. The argument {\tt arr}  indicate targeted destination of the merged result (black or white).
	
       Figure \ref{fig:insert} is an illustration of a recursive merge process. Note that, with this example,  when the new value \enquote{52} is being inserted (Figure \ref{fig:insert}(a)), the total number of values in the BWA is 7 = [0111], which indicates the lowest ranked three white segments are active. It follows that the insert will recurse three times (Figures \ref{fig:insert}(b), (c)), and (d) before it settles. At that time, we increase the state variable total by one to become 8 = [1000]. It follows that the next insertion will be done by one line of code (Line 3 of Listing \ref{listing_insert}).
       
       The merge of two segments will always leave the old values that were merged in their original places, the corresponding segment then becomes {\it inactive}. It should be noted that there is no need to clear out the values in the inactive segments, for they will always be indicated by the \enquote{0} valued bit in the variable {\tt total} after the recursion is completed, and will not be referred to until the segments are filled with new values and become active again. In Figure \ref{fig:insert}, we have used horizontal line to cross out those values that are in the inactive segments.

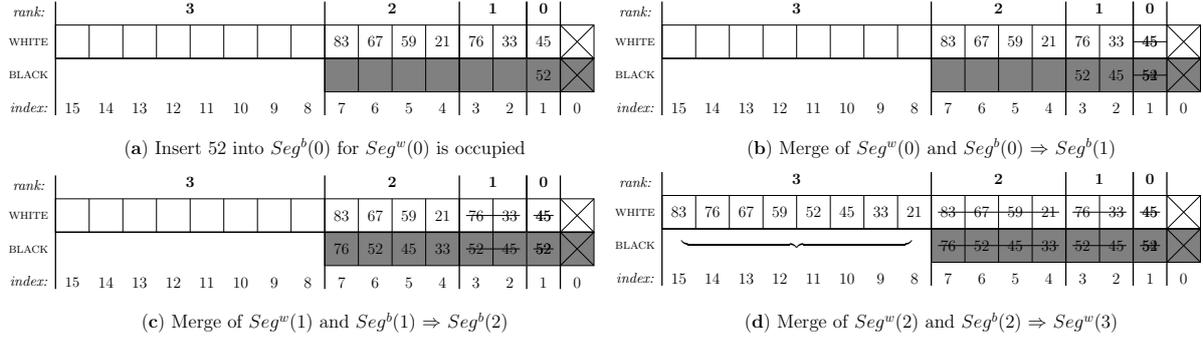
\begin{figure}[h!]
\begin{center}
%\resizebox{1.9\textwidth}{!} {      
%\input{Figures_latex/BWA_insert.tex}
%}
\resizebox{0.95\textwidth}{!} {      
\input{BWA_insert_twocolumns.tex}
}
\end{center}
\caption{The insertion of an element and the invoked recursion process}
\label{fig:insert}
\end{figure}

\begin{figure}[h]
%\begin{minipage}{\linewidth}
\begin{lstlisting}[caption={The merging procedure of segments with indices between s and t.},label=prog:insert, captionpos=b,abovecaptionskip=0.7em,mathescape=true]
 merge (s, t, A)	
      i = s   
      j = s 
      while ((i $\leq$ t) && ( j $\leq$ t))
             if B[i] $\leq$ W[j] 
                     A[k] = B[i]
                     i = i+1
             else
                      A[k] = W[j]
                      j = j+1
             end
             k = k+1
       end
       if (i >t)
            for (u = j; u <=t; u++)
                  A[k] = W[u]
                  k = k+1
            end
       else
            for (u = i; u <=t; u++)
                  A[k] = B[u]
                  k = k+1
            end
       end
end
\end{lstlisting}   
\end{figure}

%\end{minipage}

In the following, the {\em weight} of an integer refers to  the number of non-zero bits  in its binary form. For example,  the weights of  number 7 and 32 are respectively 3 and 1.  Observe that the programs in  Listing \ref{listing_insert}, \ref{prog_merge}, and \ref{prog:insert}  lead to  the following properties of the dynamics of the Insert operation:

\begin{property}[Insert]
Let the number of values stored in a BWA be {\tt total}  $ = (b_{m-1}, \ldots, b_i, \ldots, b_0) $ with weight $w$, then there are exactly $w$  active segments in the BWA; (b) a segment of rank $i$ is active if and only if $b_i$ =1; (c) values in all active segments are sorted in ascending order; (d) when ${\tt total} = 2^k$ is a $k$th power of two,  there is only one active segment with rank of $k$, and all the values in the BWA are stored in that segment.
\end{property}

A number $n$ is said to be {\it k-trailed} if it has the binary form of  $ n = [b_{m-1}, \ldots, b_{k+1}. 0, \underbrace{1, \ldots,1}_{k}]$. For instance, 
the number 7, 11, and 9 are respectively 3-, 2-, and 1-trailed. A number is {\it strongly k-trailed} if it is k-trailed, and $b_i = 0$, for all $i > k$. The number 15, e.g. is strongly 4-trailed\footnote {It is easy to show that any integer $n$ of the form $2^k-1$ is strongly $k$-trailed.}. Note that any even number is, by definition, 0-trailed.

\begin{theorem}[Recursion of Merge]
\label{theorem:insert}
Let $n$ be the number of values stored in a BWA, and it  is $k$-trailed, then an insertion of a value into the BWA with $n$ values leads to exactly  $k$ recursive merges.
\end{theorem}

\begin{corollary}[Destination of Merge]
\label{corollary:insert}
Suppose there are $n$ values in the BWA, and $n$ is $k$-trailed. When a new values is inserted, then
(a) the result  of the $k$ recursive merges will be held in the white segment of rank $k$, which include all the values in segments with ranks smaller than $k$; (b) All the segments with rank smaller than $k$ will become inactive once the insertion is completed; (c) if $n$ is strongly k-trailed, after the insertion, all the values in the BWA will be held in the segment of rank $k$. 

\end{corollary}

The variable {\tt total}, which records the number of values stored in the BWA, therefore, encoded all the information about the BWA configuration. 
When it  increases with an insertion, the changing process in its binary form  is isomorphic to the  merging dynamics of the segments.  This is an elegant as well as useful property,  which is exploited throughout the coding  of the 
BWA operations.

\subsection{Search}

The following is the program for the Search operation:

\begin{lstlisting}[caption={The Search for a value $v$ in a BWA with size $N = 2^k$},label=prog:search, captionpos=b,abovecaptionskip=0.7em,mathescape=true]
Search(v)  		
      r = Nil
      for (i=k-1; i $\geq$ 0; i--)
            if (active(i))
                r = segSearch(v, i)
                if (r != Nil)
                      break
                end
            end
       end
       return r
end
\end{lstlisting}   

Thus, the  search for a value in a BWA is reduced to the searches over active segments, starting from the active segment with highest rank\footnote{A search starting from the lower ranked to higher ranked active segments will not affect the correctness of the search operations, however, it is not as
efficient as the other way around when a value may appear multiple times in the BWA.}. If the the value is found, the index of the entry with the value is returned. It returns Nil if and only if the value is not found in any of the active segments.

The search of a segment is performed by segSeach, 

\begin{lstlisting}[caption={The search of a value within a given segment of rank $i$},label=prog:segSearch, captionpos=b,abovecaptionskip=0.7em,mathescape=true]
segSearch(v, i)  		
      s = S(i)
      t  = T(i)
      r = binarySearch(v, s, t)
      return r
end
\end{lstlisting}   

The above listing calls a  binarySearch  (Line 4) over a section of an array with indices ranging from  $s$ to $t$ inclusive. The algorithm is  just  a standard and simple binary search. For the reason of completeness, a version of the code is given in Listing \ref{prog_binarySearch}.

\begin{figure}[h]
%\begin{minipage}{\linewidth}
\begin{lstlisting}[caption={Binary search of a value $v$ between indices $s$ and $t$},label=prog_binarySearch, captionpos=b,abovecaptionskip=0.7em,mathescape=true]
binarySearch(v, s, t)  
      r = Nil		
      if ((t - s) $\leq$ 1)
          if (W[s] == v) return s
          if (W[t]  == v) return t
          return r
      else
          m = (s+t) $\gg$ 1
          if (W[m] $\leq$ v)
               r =   binarySearch(m, t, v)
          else 
               r = binarySearch(s, m, v)
          end
      end
      return r
end
\end{lstlisting}  
%\end{minipage}
\end{figure}
Note that the middle point between $s$ and $t$ is calculated at Line 8  with a binary right shift.

\subsection{Delete}

The deletion of value from BWA is straightforward by itself once a search for it is performed.  When the value is found, a special value VOID is used
to replace it. However, to prevent the deterioration of performance when the value of VOID becomes dominating in some segments, some care must
be taken. 

For this reason,  we introduce an array V, referred to as  the {\it occupancy vector} of length $k$ for a BWA of size $N = 2^k$, where $V[i]$ records the number of real (non-VOID) values of segment of rank $i$. A procedure named {\it demote} is introduced, which will be invoked whenever  the number of real values reached a threshold, which is half of the length of a given segment. 

The code for Delete is given below:

\begin{figure}[h]
\begin{lstlisting}[caption={Delete of a value $v$ from a BWA},label=prog:delete, captionpos=b,mathescape=true]
Delete(v)
      r = Search(v)
      if (r == Nil) return Nil 
      W[r] = VOID
      i = seg(r)
      V[i] = V[i] -1
      if (V[i] $\leq$ (size(i) $\gg$ 1))
             demote(j)
             if (active(j-1))
                  merge (j-1)
                  total = total - size (j-1)
                  V[j] = V[j] + V[j-1]
             else
                  total = total - size (j)
                  V(j-1) = V[j]        
             end
       end		
end	
\end{lstlisting}  
\end{figure}

As can be seen in Listing \ref{prog:delete},  {\tt delete} always starts from a search for the value (Line 2). If the value is not found, it simply returns Nil to indicate the failure of the operation (Line 3); otherwise, the index of the value is returned, and the deleted value is replaced by a VOID (Line 4). With a successful delete, the value  V[i]  is decreased by one (Line 6) to reflect the loss of a
real value in the segment in which the deletion occurred. The occupancy rate of the segment is then checked (Line 7). If
it does not reach 50\%, the deletion is completed; otherwise, a demotion of the segment to the next lower rank will take place (Line 8). Depending on whether the lower ranked segment is active, the demote will copy all the non-VOID values to either the black or the white segment of lower rank. In and only in the former case, a merge of the black and white segments will take place at the lower rank (Line 10). Lines 11, 12, 14, 15 are the code to update the state variable {\it total} and  occupancy vector $V$ of the BWA. 

Note that the call to {\tt seg} at Line 5 of the above listing returns the segment its argument index belongs to. It can be easily calculated by an examination of
the binary of the argument  index, the position of most  significant bit\footnote{Counting from right to left, starting from zero.} that equals to one is the rank the segment that the  index belongs, e.g. seg(7) = 2 and seg(11) = 3.

The code for {\tt demote} that is referred to at Line 8 in Listing \ref{prog:delete} is given below:

\begin{minipage}{\linewidth} % prevent it to bread into two pages
\begin{lstlisting}[caption={The demote  operation of a half emptied segment of rank $i$},label=prog:demote, captionpos=b,abovecaptionskip=0.7em,mathescape=true]
demote(i)
	s = S(i)
	t  = T(i)
	arr = (active(i-1)? B:W)
	j = S(i-1)
	for (j= s; j <= t; j++)
	     v = W[j]
	     if (v != VOID)
	           arr[j] = v
	     end
	     j = j+1
	end			
end
\end{lstlisting}  
\end{minipage}

Figure \ref{fig:demote} illustrates a delete of the value 59  to  take place   in  segment  of rank 3 (a);  the number of non-VOID values in
the segment after the deletion decreased to  half of its length, a demotion is then to be invoked (b); since segment of rank 2 is active, the demotion result is 
put in black segment (c); a merge at rank 2 takes place, all the values from the white and black segments of rank 2 are now in segment of rank 3 (d).
Note that the value of VOID is denoted by $\phi$ in the illustration. 

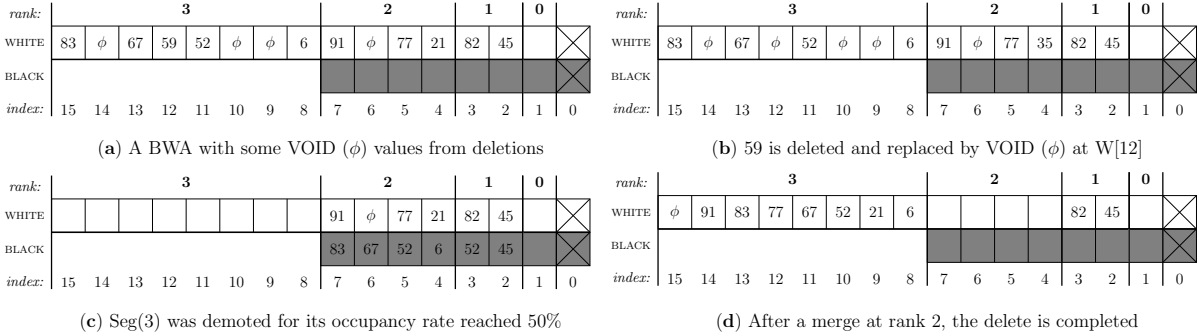
\begin{figure}[htb]
\begin{center}
\resizebox{0.95\textwidth}{!} {      
\input{BWA_delete.tex}
}
\end{center}
\caption{The demotion process invoked  by  a delete }
\label{fig:demote}
\end{figure}

The code in  Listing \ref{prog:delete} and \ref{prog:demote} leads to  the following properties:

\begin{property}[Delete]
\label{theorem:delete}
(a) The delete of a value leads to a demotion if and only if the number of non-VOID values decreased to  half of the segment size due to the deletion; (b) when a segment of rank $i$ is demoted, all its non-VOID values will be moved to segment of rank $(i-1)$; (c) when the segment of rank $(i-1)$ is inactive, the result of demotion will be put in the white segment of rank $(i-1)$; (d) otherwise, when the
segment of rank (i-1) is active, the demoted segment will be held in the black segment of rank $(i-1)$,  followed by a merge of black and white segments of rank $(i-1)$, the result is then put back to the white segment of rank  $i$.
\end{property}

In the following, the {\it occupancy rate} of a segment refers to the ratio between the number of non-VOID values in the segment and the size of the segment.

\begin{theorem}[Occupancy Rate]
A demote operation results in either a new  active segment of rank lower by one with an occupancy of 100\% if the lower ranked segment was inactive; 
or a new segment of the same rank with an occupancy rate strictly greater than 75\% after the merge with the segment of rank lower by one.
\end{theorem}

\begin{corollary}[Occupancy Lower Bound]
\label{property:delete}
At any stable state of a BWA, the occupancy rate of any active segment is always strictly greater than 50\%.
\end{corollary}

It is this  lower bound of the occupancy rate that guarantees the BWA performance will not deteriorate with arbitrary sequences of operations, including
deletes.

\section{Performance}
\label{sec:performance}

The dynamics and the complexities of the BWA operations are discussed Section \ref{subsec:analysis}. In Section \ref{subsec:benchmark}, a C++ implementation of BWA is described, the performance data of the BWA in C++ are presented and plotted, serving in part as  a validation to the theoretical analysis.

\subsection{Analysis}
\label{subsec:analysis}

When a sequence of values are inserted into a BWA, the merging process can be illustrated with a {\it BWA merge tree}. 
For a BWA of size $N = 2^{k+1}$, the merge  tree has $k$ levels, corresponding to the different ranks of the segments. The leaf nodes correspond to the input values from left to right in the order they are inserted. An internal node represents a merge between
the two segments denoted by its children. A white colored node represents a white segment, and a black colored node 
represents a black segment. Figure \ref{fig:bwa_tree} illustrates such a merge tree with four levels.

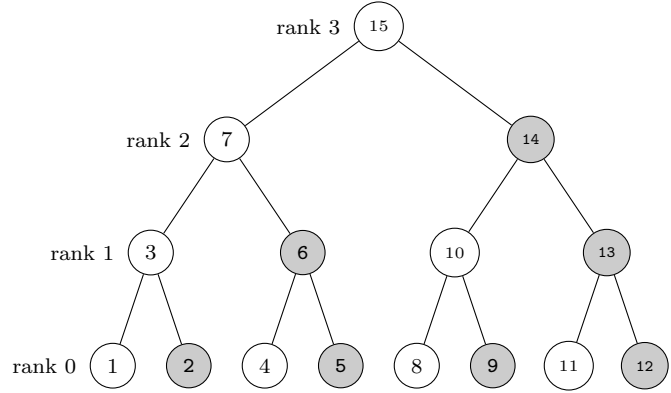
\begin{figure}[hbt]

\begin{center}
%\resizebox{0.95\textwidth}{!} {    
\input{BWATree.tex}
%}
%\vspace{-6mm}
\end{center}
\caption{The BWA merge tree and its post order traversal}
\label{fig:bwa_tree}
\end{figure}

The insert sequence in BWA leads to the following simple dynamics of the tree: (a) the input values are put to the leaf nodes, one at a time; (b) for any internal node, a merge of its two children of different colors is invoked whenever  both of its children are present\footnote{This can also be understood  from the view of a data flow paradigm: an operation is invoked whenever its arguments are ready.}. 

A close examination of the BWA merge, then,  leads to

\begin{lemma}[BWA Merge Tree]
\label{lemma:post_order}
The order of the merges represented by the nodes in a BWA merge tree is a post order traversal of the tree.

\end{lemma}
\begin{proof}
Obvious. 
\end{proof}

In Figure \ref{fig:bwa_tree},
the labels inside the nodes of the BWA merge tree %in Figure \ref{fig:bwa_tree}
have been used to indicate their sequencing number  in a post-order traversal of the tree.

It is interesting to compare the BWA merge tree with that of  mergesort. Figure \ref{fig:merge_tree} is an illustration of
the merge tree of mergesort for eight values. It is obvious that  the merge trees for BWA and mergesort share the same tree structure, the difference lies only in the order the nodes are traversed, which represent the sequence of merge operations that an algorithm performs\footnote{Hence, BWA can be said to perform \enquote{incremental merge sort}  with exactly the same complexity as merge sort, together with some other functionalities that merge sort does not provide.}.

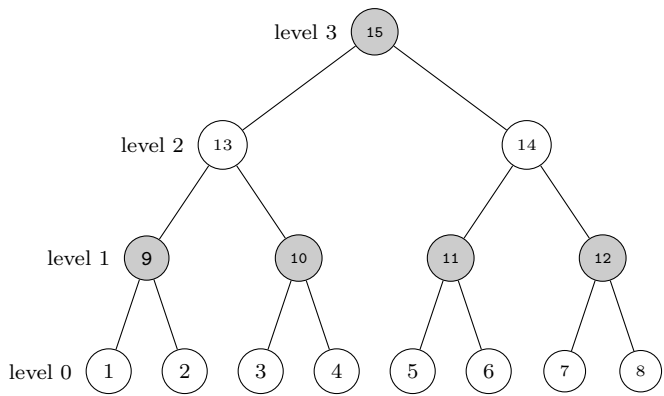
\begin{figure}[hbt]
\input{MergesortTree.tex}
\vspace{-6mm}
\caption{The operational order in mergesort}
\label{fig:merge_tree}
\end{figure}

The traversal order of the merge tree for mergesort is one level at a time, from the leaf level to the root.
The internal nodes, which represent  merge operations,  at a higher level,   are not to be visited until all the nodes of the lower level have been visited. An illustration of the merge tree
for {\tt mergesort}   is given  in Figure \ref {fig:merge_tree}  with the labelled order by which the nodes are  travesesed.

\begin{theorem}[Insert Time]
The insertion of a sequence of values of length $n=2^k$ into a BWA takes $O(n\log(n))$ time.
\label{theorem:insertTime}
\end{theorem}
\begin{proof}
There are $\frac{n}{2^(i+1)}$ merges for segments of rank $i $ with length $2^i$, for $i = 0$ to $(k-1)$, therefore the total number of comparisons, T(n), is given by
\[   
    T(n) = \Sigma_{i=0}^{k-1}\ \frac{n}{2^{i+1}}\times 2^i = O(n\log(n))
\] \qedhere
\end{proof}

\begin{corollary}[Amortized Cost of Insert]
Insertion of  $n$ values into a  BWA  takes  average  $O(\log(n))$ time per insert. 
\label{corollary:insert}
\end{corollary}

This conclusion, while seemingly trivial,  is actually interesting in that it holds despite the fact that one can  infer from  the code in the Listing \ref{listing_insert} that the time taken
by an insert is $O(1)$ every other time an insertion is performed.

A search is said to be a {\it hit} if the searched value  is found in the BWA, otherwise, a {\it miss}. There can be some difference between
the time spent on hit and miss under certain conditions. In the following, we will make a distinction between hit and miss in some cases.

The search time obviously depends on the number and the ranks of the active segments in a BWA at the time of the search.  We will start the analysis with a simple case:

\begin{theorem}[Search Time, case $n=2^m$]
\label{theorem:search_perfect}
When the number of values in a BWA is a power of two of the form $n = 2^m$, the time to search for  a value in the BWA is $O(\log (n))$.

\end{theorem}
\begin{proof}
In this case, all the values are in segment of rank $m$ of length $n$. %\qedhere
\end{proof}
Theorem \ref{theorem:search_perfect}  holds true regardless whether the search is a hit or miss under the given condition. However, when $n$ is not a power of two, more than  one segment may need to  be searched. 
The worst case occurs when $n = 2^m-1$, and in this case, every segments of rank smaller than $m$ is active, and as such, could be potentially searched\footnote{We say potentially searched, for the search could be terminated before a segment is searched for the value was found in a segment of higher rank.}.

\begin{theorem}[Search Time (Hit)]
\label{theorem:search}
For a BWA with up to   $ n = 2^m $ values, the amortized  time of  a hit search is $O(\log n)$ , provided that the values are uniformly distributed.
\end{theorem}
\begin{proof}
The probability that a values falls in a given segment of BWA is proportional to the length of the segment under the uniform distribution assumption.
The search algorithm starts the search from active segment with highest rank, which has a length half of the size of the BWA, and hence the probability of finding the value in it is $\frac{1}{2}$. The probability of finding the value in the next lower rank segment is $\frac{1}{4}$,  and, generally, the probability of finding the value in a segment of rank $i$ is half of that of rank  $(i+1)$. Note that finding a value in segment $i$ means that  all  the  segments of higher ranks have been searched,  so the cost for searching the higher ranked segments should also be charged together with the time taken to search the current segment.  The search is performed only over active segments, the worst case would be all the segments are active\footnote{This worst case
occurs  when $n = 2^k -1$, and in this case, all the least significant bits take the value of one.}, and in that case each segment will be searched until the value is found. 

It follows that the expected time  for a hit search in a BWA with $n$ stored values is 
\scriptsize
\begin{eqnarray*}
\tiny
 T(n)  &=& \frac{1}{2} \log(2^{m-1} ) + \frac{1}{4} (\log(2^{m-1})  + \log(2^{m-2}))+ \frac{1}{8} (\log(2^{m-1} )+ \log(2^{m-2}) + \log(2^{m-3})) + \cdots 
          + \frac{1}{2^{m-1}}\Sigma_{i=1}^{m-1} \log(2^{m-i}) \\
         &=& (\frac{1}{2} + \frac{1}{4} + \frac{1}{8} + \cdots + \frac{1}{2^m}) \log (2^{m-1})  \\
           &&    \hspace{0.5cm}         + (\frac{1}{4} + \frac{1}{8} + \frac{1}{16} +\ldots + \frac{1}{2^m})  \log (2^{m-2})  \\
           && \hspace{1.2cm}                + (\frac{1}{8} + \frac{1}{16} + \frac{1}{32} \cdots + \frac{1}{2^m}) \log (2^{m-3})+  \\
           & &   \hspace{1.9cm}               \ddots  \\
            & &      \hspace{2.5cm}   + \frac{1}{2^{m-1}} \log(2)        \\
          &\leq&  (m-1) + \frac{1}{2} (m-2) + \frac{1}{4} (m-3) + \cdots, + \frac{1}{2^{m-1}} \\
          &\leq & 2 (m-1) \\
          & =& O(\log(n) ) \qedhere
\end{eqnarray*}

\end{proof}

\begin{theorem}[Search Time (Miss)]
The search for a value that is not present in the BWA with up to $n = 2^m-1$ values takes $O(\log^2 (n))$  average time.
\label{theorem:miss}
\end{theorem}
\begin{proof}
There are n different possible configurations of the BWA. A segment of rank $i$  of length $2^i$ is active in half of the configurations. The total  time of the $n$ searches over all the different  configurations, divided by the total of $n$,  is then 
\begin{eqnarray*}
\tiny
          T(n) &=&   \frac{\frac{n}{2} (\log 2^{m-1} + \log 2^{m-2} + \cdots  +1)}{n}\\
                 &=& O(\log^2(n) )
\end{eqnarray*}
\end{proof}

It is obvious from Listing \ref{prog:delete} that the deletion of  an existing value in the BWA takes time no more than a search provided that the occupancy rate of the segment where the deletion takes place does not decrease to 50\% as a result of the operation. On the other hand, If a deletion leads the
occupancy rate to drop to the 50\% threshold, a demotion is kicked in, which may and may not be followed by a merge, in both cases, the time taken will be proportional to the length of the segment.

\begin{theorem}[Amortized Delete Time]
\label{theorem:delete}
The amortized cost of deleting a  value in the BWA is $O(\log(n))$.
\end{theorem}
\begin{proof}
Consider a segment S of rank $k$, and hence its  length is  $n = 2^k$, half of the values are deleted one by one without demotion. We know that for each of the deletions, a 
search is required, which takes $O(\log(n))$ time. 
The last one will take time proportional to $n$ since it will invoke the demotion process, and possibly followed by a merge of rank $(k-1)$. The amortized time of the deletions from the time the occupancy rate is one to the time when the demotion is completed, then, is

\scriptsize
\begin{eqnarray*}
 T(n)= \frac{{O(\frac{n}{2} \log(n))  +O( n)}}{\frac{n}{2}}= O(\log(n))
\end{eqnarray*}
\end{proof}

The above  theorem showed that the delete operations  does not bring   the amortized asymptotic complexity of the operation to go beyond  that   of a search for the corresponding value that is  present in the BWA.  Together with Theorem \ref{theorem:search_perfect}, \ref{theorem:search}, \ref{theorem:miss}, we may conclude that a  sequence  of operations with mixed  search and delete operations with hit ratio\footnote{Referring  to the ratio of the number of hits among the total number of operations.} between 1 and 0, will demonstrate amortized complexity somewhere 
between $O(\log(n))$ and $O(log^2(n))$. Therefore, a higher probability for the searched or deleted values to be  present in the data set generally leads to a lower amortized cost of the operations.

Finally, let us examine the {\it space efficiency} of the BWA, which is defined as  the ratio between the space used by the data and the total space used. 
In a BWA, the black array is used as  the scratch space for the BWA operations.
The space taken by the black array is half
of that  used by the white  array, which is fully used by the real data.   The space efficiency of BWA is therefore $\frac{1}{1.5}  = \frac{2}{3} \approx 0.667$,  which is  much higher than that of any tree based data structures\footnote{In tree based dynamic date structures,  each node of the tree has to hold pointers to its children and parent as well as a piece of data.} for dynamic data sets.

\subsection{Implementation and Performance}
\label{subsec:benchmark}

The BWA has been implemented in  C++ as a class. The size of the BWA $N$, the logarithmic $k$ of $N$, the occupancy vector $V$, where $V[i]$ shows the 
number of non-VOID values  in segment of rank $i$, are all private members of the class. The operations, Insert, Search, Delete,  are made public, whereas merge, demote, and some other auxiliary  functions are private. The programs in C++  follow closely in structure and style of the code   presented  in Section \ref{sec:operations}.

The testing of the performance  was conducted on a Mac Pro laptop computer running Mojave MacOS. The system comes  with a i7 processor with   six cores, clocked  at 2.6 GHz. It has  9 MB L3 cache, and 256 KB L2 cache per core. At the time to start the testing, the free memory is maintained at  the level between four to six GB, whereas  the CPU is about 90\%  idle.  The values used in the tests are randomly generated with a range beyond  the size  of  the BWA. The chrono time library \cite{chrono}, which comes with C++ 11, is used for the timing.

The tests are conducted for BWAs with size N  ranges between $2^{10} = 1,024$ and $2^{29} =536,870,912$. 
The testing values are four byte integers. Hence, the space that the BWA itself takes up during the tests will go up to 3 GB. The random values used during the  tests are pre-generated and being held an array of size up to 2 GB. The Performance Monitor of the Mac system indicated that  the testing program, as an application, took around 99\% of the computing  time of one core, and 10 GB of system memory.

\begin{table}
\tiny
\begin{center}
 
 \begin{tabular}{cc}

\begin{tabular}{l l l l | l l }
\hline\hline
Size & Insert & Search & Delete & Search* & Delete*  \\ [0.5ex]
\hline
$2^{10}$ & 105.949 &142.289   &239.594 & \ 540.43 & \ 576.862\\
$2^{11}$ &\ 99.254     &139.51   & 254.676   & \ 619.694& \ 660.022\\
$2^{12}$ &110.299  &149.991  & 271.268  & \ 710.861& \ 771.738\\
$2^{13}$ & 119.603    & 158.841   & 292.494 & \ 788.732& \ 881.442\\
$2^{14}$ & 128.13      & 162.641 & 324.852  & \ 880.209& \ 954.138\\ 
$2^{15}$ & 135.945    & 173.694 & 350.686 & \ 994.758& \ 1091.76 \\
$2^{16}$ & 144.256  & 185.698   & 390.012 & 1155.82 & 1244.09\\
$2^{17}$ & 152.413  & 215.686   & 461.078 & 1293.16 & 1391.76\\
$2^{18}$ & 160.31  & 222.679 & 478.018 & 1461.91& 1497.18\\
$2^{19}$ & 162.527      & 239.829  & 546.856 & 1640.62& 1735.51\\
\hline 
\end{tabular}
&
\begin{tabular}{l l l l | l l }
\hline\hline
Size & Insert & Search & Delete & Search* & Delete*  \\ [0.5ex]
\hline
$2^{20}$ &182.601 &\ 296.176 & \ 611.42 & 1845.1 & 1955.55\\
$2^{21}$ & 193.025  &\ 373.591   & \ 664.332  & 2026.22 & 2211.12\\
$2^{22}$ & 200.564  &\ 540.384   & \ 730.102  & 2350.39 & 2485.58\\
$2^{23}$ & 211.788   &\  644.83   & \ 789.652 & 2717.32 & 2823.35\\
$2^{24}$ & 213.38   &\ 724.457 & \ 851.47 &  3062.46 &  3249.46\\ 
$2^{25}$ & 222.911    &\ 799.269 &\  911.798 & 3542.61 & 3631.58\\
$2^{26}$ & 235.683  &\ 876.615  &\  944.468 & 4056.0 & 4110.58\\
$2^{27}$ & 239.855 &\ 968.382  & 1061.95& 4702.94 & 4593.06\\
$2^{28}$ & 249.329  & 1087.05 & 1159.89 & 4886.36 & 5231.23\\
$2^{29}$ & 263.82     & 1246.46   & 1145.87& 4960.39 & 5720.49\\
\hline 
\end{tabular}

\end{tabular}
 
\end{center}
\label{table:timing}
\caption{Amortized time (ns) of operations with perfect/random(*) BWA configurations}

\end{table}

The performance data for Insert, Search, and Delete operations gathered during the tests are  listed in Table \ref{table:timing} and plotted in Figure \ref{fig:timing}.
The horizontal axis, representing  the size of the BWA, is 
scaled logarithmically, whereas the vertical axis, showing the time that operations take in the unit of nanosecond, is linearly scaled. Hence, a logarithmic function will show up as a straight line in such a figure.

The tests are conducted for two different cases. In the first case, the  state variable  {\tt total}  is a perfect power of  two, and in the other, a random number ranging up to the size of the BWA. It follows that the number of active segment  equals  one in the former case, and  equals to the weight of the random number in the latter case. For the latter case,  the average time of  a thousand  runs is taken, each of which is on a different
configuration of the BWA as captured by the state variable {\tt total}. The left and right parts of Table \ref{table:timing} showed the performance for the two cases respectively, which are in turn respectively plotted in Figure \ref{fig:timing}(a) and \ref{fig:timing}(b).

A  comparison of  Figure \ref{fig:timing}(a) and \ref{fig:timing}(b) shows that the Search and Delete operations over random configurations indeed cost more than that over the \enquote{ideal} configurations when 
the state variable {\tt total} happens to be a power of two.  In contrast, the Insert operation is insensitive to the difference of configurations, therefore, its  timings are identical in  the two figures. This should come as no surprise for it is simply a manifestation of  Corollary \ref{corollary:insert}.

\begin{figure}[h]

\begin{center}
\begin{tabular}{cc}

% generated by ~udisk/projectsC++/plots/BWA6.4B.plt
%\includegraphics[scale=0.47]{BWA6_4B.png} &
\includegraphics[scale=0.47]{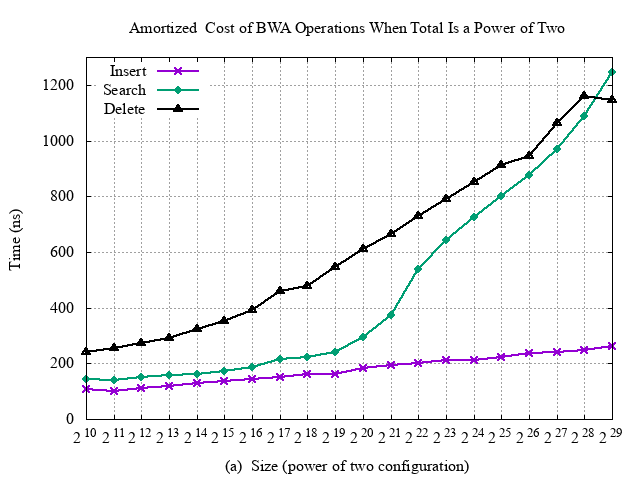} &
\hspace{-1.5em}
% generated by ~udisk/projectsC++/plots/BWA6.4C.plt
%\includegraphics[scale=0.47]{BWA6_4C.png} 
\includegraphics[scale=0.47]{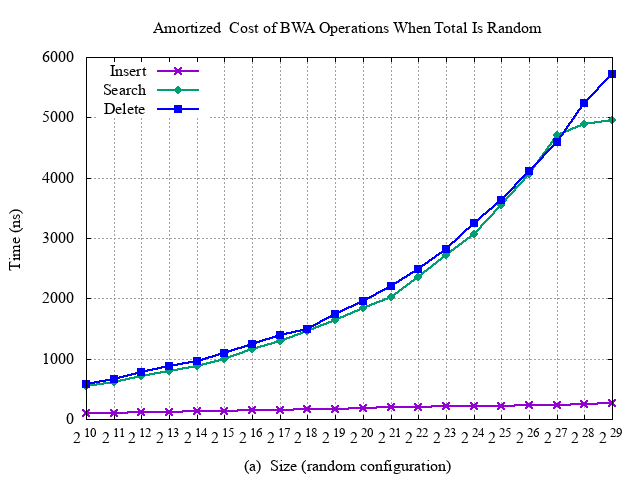} 
\end{tabular}
\end{center}

\caption{Insert, Search  and Delete time per operation (C++ on MacPro)}
\label{fig:timing}
\end{figure}

As shown in Table \ref{table:timing} and Figure  \ref{fig:timing}, the amortized time  taken by the operations in a BWA up to the size of $2^{29}$ is from 263.82 to  5720.49 nanoseconds in all the cases.

Observe that whenever the number of inserted values is a $k$-th power of two,  all the inserted value will appear in the BWA as one sorted sequence 
in the segment of rank $k$. It is therefore sensible to compare the performance data with other sorting algorithms. As pointed out in Lemma \ref{lemma:post_order},  and shown in Figure \ref{fig:bwa_tree} and \ref{fig:merge_tree}, the BWA merge tree  is isomorphic to that of  
mergesort with only a differences in the order the nodes are traversed. It follows that  the Insert of BWA is at least as fast as  that in mergesort 
per operation. Hence it should come as no surprise to see that the amortized cost went from 105.949 to  only 263.82 nanoseconds as
the number of inserted values increased  $2^{19}$  folds from 1,024 to 536,870,912 with  a fairly flat slop  as shown in Figure \ref{fig:timing}.

\section{Augmentations}
\label{sec:augmentations}

We have  shown that  black-bhite  array can be deployed as a data structure for dynamic set with Insert, Search and Delete operations. In this section, we outline some fairly straightforward  augmentations  to its functionality,  including

\begin{itemize}
   \item {\tt Max(Min)} : returns the maximum (minimum) value;
   \item {\tt ExtractMax(Min)}: Find and delete the maximum (minimum) value;
   \item {\tt LowerBound} ({\tt UpperBound}) ($v$): finds  the position  of the smallest  (largest) number which is greater (smaller) than value $v$;  
   \item {\tt Interval} $(s, t)$: Returns all the values falling between the interval of $[s, t]$;
   \item {\tt IncrementalSort}: incrementally sorts a sequence of unordered values, one at a time.
 \end{itemize}
 
 It is the property of a BWA that the maximum (minimum) value in a segment of rank $m$ is always at the top (bottom)\footnote{Meaning the position with the highest (the lowest) index.} of the segment with index of $(2^{m+1} -1) (2^m)$ provided
 that is not a VOID entry, otherwise it is to be found at the index closest to the top (bottom)  on the left (right) where the entry is not a VOID. To find the global maximum, we just need to go through all the active segments,  compare and throw out the smaller value each time a new segment is examined.  It follows that Max, Min,  ExtractMax, and ExtractMin operations can all be done in  $O(\log(n))$ time\footnote{When the number of values in the BWA, or in other words,  the state variable {\tt total} is a power of 2, the time will become O(1).}. In applications such as $k$ nearest neighbors \cite{knn}, the ExtractMin operations can be applied $k$ times.

 The LowerBound (UpperBound) can be achieved with  fairly simple modifications of the Listings \ref{prog:search}, \ref{prog:segSearch} and \ref{prog_binarySearch}. The results will be the LowerBound (s) and UpperBound (t) indices in each of the active segments, if any\footnote{The values in a segment may and may not overlap with the interval [s, t].}. Interval (s,t) can then be applied to pull out all the values in between. 
 
 We have shown that a sequence of length $n$ can be inserted one by one in  $O(n \log  (n))$ time (Section \ref{subsec:analysis}).  BWA can then be  used as an  incremental sorting algorithm at least as efficient as a well-implemented mergesort in time and space. 
 
 There are other augmented operations for BWA with practical applications, which will be discussed in separate works.

\section{Conclusion}
\label{sec:conclusion}

The black-white array (BWA) is presented as an array based data structure for dynamic data set  supporting effectively and efficiently the insert, search, and delete operations, with mathematically elegant structure, easy-to-understand code, and competitive performance against those based on linked structures such as lists and trees.  

The performance data of a C++ implementation of the BWA were presented, which validated the formal analysis of the amortized costs of the operations. The  nanosecond  and low microsecond wall-clock time per operation  suggests that  the random access nature of the underlying array structure  indeed can minimize the constant factor in front of the predicated asymptotic functions. When it comes to space utilization, BWA is far more efficient than any of  those list or tree based approaches by a wide margin.

The versatility of the BWA structure also allows  it to easily support augmented operations such as Max,  Min,  Interval, and  IncrementalSort. 
Hence, it subsumes in functionality some other data structures and algorithms e.g. priority queue and  mergesort. As an alternative to those weill-known list or tree based data structures for dynamic data sets such as skip lists \cite{redis,skiplist} and red-black tree \cite{aho74}, BWA can be deployed in a broad range of applications. One of the use cases,  which in part motivated this work, is  to apply it to point cloud registration with Iterative Closest Point (ICP) method as described  in \cite{mou-201b,mou-202}.

\bibliography{mypub.bib,bib20.bib}
\bibliographystyle{plain}

\end{document}

%% file: BWA_layout.tex
%\begin{figure}

%\begin{center}
\begin{tikzpicture}[scale=0.99]

% this the upper rectangle and lower ractangular

\draw (0,4) rectangle (16,5);  % White 
\draw [fill=gray] (8,3) rectangle (16,4); % Black

\draw[line width=0.4mm, black] (0,4) -- (16,4);  % separator

% these are vertical lines separating the bits:

\draw(1,4) -- (1,5);
\draw(2,4) -- (2,5);
\draw(3,4) -- (3,5);
\draw(4,4) -- (4,5);

%\draw(4, 3.8) -- (4,5.3);% separator between bit areas

\draw(5,4) -- (5,5);
\draw(6,4) -- (6,5);
\draw(7,4) -- (7,5);
\draw (8,4) -- (8,5);

%\draw(8,3.8) -- (8,5.3);  % separator

\draw(9,4) -- (9,5);
\draw(10,4) -- (10,5);
\draw(11,4) -- (11,5);
\draw(12, 4) -- (12,5.0);
\draw(13, 4) -- (13,5.0);

% following new: separators between segments:

%\draw[line width=0.35mm] (15,3) -- (15,5.7);  % separator
%\draw[line width=0.35mm,white] (15,3) -- (15,4);  % separator
\draw[line width=0.35mm] (15,2.3) -- (15,5.7);  % separator

\draw[line width = 0.35mm](14,3) -- (14,5.7);  % separator
\draw[line width=0.35mm,white] (14,3) -- (14,4);  % separator
\draw[line width = 0.35mm](14,2.3) -- (14,5.7);  % separator

\draw[line width=0.35mm](12,2.3) -- (12,5.7);  % separator

\draw[line width=0.35mm](8,2.3) -- (8,5.7);  % separator

\draw[line width=0.35mm](0,2.3) -- (0,5.7);  % separator

%----------------------------------------------- Notational Labels:
\node at(-0.8, 5.4) {\it rank:};

\node at(-0.8, 2.55) {\it index:};

\node at(-0.8, 4.5) {\scriptsize WHITE};
\node at(-0.8, 3.5) { \scriptsize BLACK};
\node at (4, 5.5) {\large \bf 3};
\node at (10, 5.5) {\large \bf 2};
\node at (13, 5.5) {\large \bf 1};
\node at (14.5, 5.5) {\large \bf 0};

% -------------------  ending cell cross:
\draw(15,4) -- (16,5);
\draw(15,5) -- (16,4);
\draw(15,3) -- (16,4);
\draw(15,4) -- (16,3);

% ---------------------these are the index values
\node at (0.5,2.5) {\large $15$};
\node at (1.5,2.5) {\large $14$};
\node at (2.5,2.5) {\large $13$};
\node at (3.5,2.5) {\large $12$};

\node at (4.5,2.5) {\large $11$};
\node at (5.5,2.5) {\large $10$};
\node at (6.5,2.5) {\large $9$};
\node at (7.5,2.5) {\large $8$};

\node at (8.5,2.5) {\large $7$};
\node at (9.5,2.5) {\large $6$};
\node at (10.5,2.5) {\large $5$};
\node at (11.5,2.5) {\large $4$};

\node at (12.5,2.5) {\large $3$};
\node at (13.5,2.5) {\large $2$};
\node at (14.5,2.5) {\large $1$};
\node at (15.5,2.5) {\large $0$};

% Cell separators for lower rectangular, (in fact can be merged with the upper one)
\draw(9,3) -- (9,4);
\draw(10,3) -- (10,4);
\draw(11,3) -- (11,4);
\draw(12, 3) -- (12,4);
\draw(13, 3) -- (13,4);

% the following 6 lines are for patent filing labels
%\draw[line width=1] (4.4,5.5) -- (5.1,5.9) -- (5.3, 5.9);
%\node at (5.8, 5.9) {801};

%\draw[line width=1] (0.75,3.7) -- (1.5,3.3) -- (1.7,  3.3);
%\node at (2.2, 3.3) {802};

%\draw[line width=1] (7.3,3.7) -- (6.7,3.3) -- (6.5,  3.3);
%\node at (6.0, 3.3) {803};

\node at(8, 1.3) {\scriptsize  A White-Black Array with R=4, N=16,};

\end{tikzpicture}

%\end{center}

%\caption{The layout of a White-Black Array of Rank 4}
%\label{fig:layout}
%\end{figure}

%% file: BWA_insert_twocolumns.tex
%\begin{figure}

%\resizebox{0.5\textwidth}{!}
%\begin{tikzpicture}[scale=0.9]

\begin{tabular}{cc}
\begin{tikzpicture}[scale=0.9]

% this the upper rectangle and lower ractangular

\draw (0,4) rectangle (16,5);  % White 
\draw [fill=gray] (8,3) rectangle (16,4); % Black

\draw[line width=0.4mm, black] (0,4) -- (16,4);  % separator

% these are vertical lines separating the bits:

\draw(1,4) -- (1,5);
\draw(2,4) -- (2,5);
\draw(3,4) -- (3,5);
\draw(4,4) -- (4,5);

%\draw(4, 3.8) -- (4,5.3);% separator between bit areas

\draw(5,4) -- (5,5);
\draw(6,4) -- (6,5);
\draw(7,4) -- (7,5);
\draw (8,4) -- (8,5);

%\draw(8,3.8) -- (8,5.3);  % separator

\draw(9,4) -- (9,5);
\draw(10,4) -- (10,5);
\draw(11,4) -- (11,5);
\draw(12, 4) -- (12,5.0);
\draw(13, 4) -- (13,5.0);

% following new: separators between segments:

%\draw[line width=0.35mm] (15,3) -- (15,5.7);  % separator
%\draw[line width=0.35mm,white] (15,3) -- (15,4);  % separator
\draw[line width=0.35mm] (15,2.3) -- (15,5.7);  % separator

\draw[line width = 0.35mm](14,3) -- (14,5.7);  % separator
\draw[line width=0.35mm,white] (14,3) -- (14,4);  % separator
\draw[line width = 0.35mm](14,2.3) -- (14,5.7);  % separator

\draw[line width=0.35mm](12,2.3) -- (12,5.7);  % separator

\draw[line width=0.35mm](8,2.3) -- (8,5.7);  % separator

\draw[line width=0.35mm](0,2.3) -- (0,5.7);  % separator

%-------------------------------------------------------- Values in the cell

\node at (8.5,4.5) {\large$83$};
\node at (9.5,4.5) {\large $67$};
\node at (10.5,4.5) {\large $59$};
\node at (11.5,4.5) {\large $21$};
\node at (12.5,4.5) {\large $76$};
\node at (13.5,4.5) {\large $33$};
\node at (14.5,4.5) {\large $45$};
\node at (14.5,3.5){\large $52$};
%----------------------------------------------- Notational Labels:
\node at(-0.8, 5.4) {\it rank:};

\node at(-0.8, 2.55) {\it index:};

\node at(-0.8, 4.5) {\scriptsize WHITE};
\node at(-0.8, 3.5) {\scriptsize BLACK};
\node at (4, 5.5) {\large \bf 3};
\node at (10, 5.5) {\large \bf 2};
\node at (13, 5.5) {\large \bf 1};
\node at (14.5, 5.5) {\large \bf 0};

% -------------------  ending cell cross:
\draw(15,4) -- (16,5);
\draw(15,5) -- (16,4);
\draw(15,3) -- (16,4);
\draw(15,4) -- (16,3);

% ---------------------these are the index values
\node at (0.5,2.5) {\large$15$};
\node at (1.5,2.5) {\large $14$};
\node at (2.5,2.5) {\large $13$};
\node at (3.5,2.5) {\large $12$};

\node at (4.5,2.5) {\large $11$};
\node at (5.5,2.5) {\large $10$};
\node at (6.5,2.5) {\large $9$};
\node at (7.5,2.5) {\large $8$};

\node at (8.5,2.5) {\large $7$};
\node at (9.5,2.5) {\large $6$};
\node at (10.5,2.5) {\large $5$};
\node at (11.5,2.5) {\large $4$};

\node at (12.5,2.5) {\large $3$};
\node at (13.5,2.5) {\large $2$};
\node at (14.5,2.5) {\large $1$};
\node at (15.5,2.5) {\large $0$};

% Cell separators for lower rectangular, (in fact can be merged with the upper one)
\draw(9,3) -- (9,4);
\draw(10,3) -- (10,4);
\draw(11,3) -- (11,4);
\draw(12, 3) -- (12,4);
\draw(13, 3) -- (13,4);

% the following are the labels for the patent filing
%\draw[line width=1](14.5,2.2) -- (15,1.7) -- (15.2, 1.7);
%\node at (15.6, 1.7) {901};

\node at(8, 1.3) {\Large ({\bf a}) Insert 52 into $Seg^b(0)$ for $Seg^w(0)$  is occupied};

\end{tikzpicture}

&

\begin{tikzpicture}[scale=0.9]

% this the upper rectangle and lower ractangular

\draw (0,4) rectangle (16,5);  % White 
\draw [fill=gray] (8,3) rectangle (16,4); % Black

\draw[line width=0.4mm, black] (0,4) -- (16,4);  % separator

% these are vertical lines separating the bits:

\draw(1,4) -- (1,5);
\draw(2,4) -- (2,5);
\draw(3,4) -- (3,5);
\draw(4,4) -- (4,5);

%\draw(4, 3.8) -- (4,5.3);% separator between bit areas

\draw(5,4) -- (5,5);
\draw(6,4) -- (6,5);
\draw(7,4) -- (7,5);
\draw (8,4) -- (8,5);

%\draw(8,3.8) -- (8,5.3);  % separator

\draw(9,4) -- (9,5);
\draw(10,4) -- (10,5);
\draw(11,4) -- (11,5);
\draw(12, 4) -- (12,5.0);
\draw(13, 4) -- (13,5.0);

% following new: separators between segments:

%\draw[line width=0.35mm] (15,3) -- (15,5.7);  % separator
%\draw[line width=0.35mm,white] (15,3) -- (15,4);  % separator
\draw[line width=0.35mm] (15,2.3) -- (15,5.7);  % separator

\draw[line width = 0.35mm](14,3) -- (14,5.7);  % separator
\draw[line width=0.35mm,white] (14,3) -- (14,4);  % separator
\draw[line width = 0.35mm](14,2.3) -- (14,5.7);  % separator

\draw[line width=0.35mm](12,2.3) -- (12,5.7);  % separator

\draw[line width=0.35mm](8,2.3) -- (8,5.7);  % separator

\draw[line width=0.35mm](0,2.3) -- (0,5.7);  % separator

%-------------------------------------------------------- Values in the cell

\node at (8.5,4.5) {\large $83$};
\node at (9.5,4.5) {\large $67$};
\node at (10.5,4.5) {\large $59$};
\node at (11.5,4.5) {\large $21$};
\node at (12.5,4.5) {\large $76$};
\node at (13.5,4.5) {\large $33$};
\node at (14.5,4.5) {\large $45$};
\node at (14.5,4.5) {\large $45$};
\node at (14.5,3.5){\large $54$};
\draw(14,3.5)-- (15,3.5);
\node at (14.5,3.5){\large $52$};
\draw(14,4.5)-- (15,4.5);
\node at (13.5,3.5) {\large $45$};
\node at (12.5,3.5){\large  $52$};
%----------------------------------------------- Notational Labels:
\node at(-0.8, 5.4) {\it  rank:};

\node at(-0.8, 2.55) {\it index:};

\node at(-0.8, 4.5) {\scriptsize WHITE};
\node at(-0.8, 3.5) { \scriptsize BLACK};
\node at (4, 5.5) {\large \bf 3};
\node at (10, 5.5) {\large \bf 2};
\node at (13, 5.5) {\large \bf 1};
\node at (14.5, 5.5) {\large \bf 0};

% -------------------  ending cell cross:
\draw(15,4) -- (16,5);
\draw(15,5) -- (16,4);
\draw(15,3) -- (16,4);
\draw(15,4) -- (16,3);

% ---------------------these are the index values
\node at (0.5,2.5) {\large$15$};
\node at (1.5,2.5) {\large $14$};
\node at (2.5,2.5) {\large $13$};
\node at (3.5,2.5) {\large $12$};

\node at (4.5,2.5) {\large $11$};
\node at (5.5,2.5) {\large $10$};
\node at (6.5,2.5) {\large $9$};
\node at (7.5,2.5) {\large $8$};

\node at (8.5,2.5) {\large $7$};
\node at (9.5,2.5) {\large $6$};
\node at (10.5,2.5) {\large $5$};
\node at (11.5,2.5) {\large $4$};

\node at (12.5,2.5) {\large $3$};
\node at (13.5,2.5) {\large $2$};
\node at (14.5,2.5) {\large $1$};
\node at (15.5,2.5) {\large $0$};

% Cell separators for lower rectangular, (in fact can be merged with the upper one)
\draw(9,3) -- (9,4);
\draw(10,3) -- (10,4);
\draw(11,3) -- (11,4);
\draw(12, 3) -- (12,4);
\draw(13, 3) -- (13,4);

% patent filing labels
%\draw[line width=1](13,2.2) -- (13.5,1.7) -- (13.7, 1.7);
%\node at (14.2, 1.7) {902};

\node at(8, 1.3) {\Large ({\bf b}) Merge of $Seg^{w} (0) $ and $Seg^b (0)$   $\Rightarrow$ $Seg^b (1) $};

\end{tikzpicture}

\\

%----------------------------------------------------------------------------------------------------------------------------
\begin{tikzpicture}[scale=0.9]

% this the upper rectangle and lower ractangular

\draw (0,4) rectangle (16,5);  % White 
\draw [fill=gray] (8,3) rectangle (16,4); % Black

\draw[line width=0.4mm, black] (0,4) -- (16,4);  % separator

% these are vertical lines separating the bits:

\draw(1,4) -- (1,5);
\draw(2,4) -- (2,5);
\draw(3,4) -- (3,5);
\draw(4,4) -- (4,5);

%\draw(4, 3.8) -- (4,5.3);% separator between bit areas

\draw(5,4) -- (5,5);
\draw(6,4) -- (6,5);
\draw(7,4) -- (7,5);
\draw (8,4) -- (8,5);

%\draw(8,3.8) -- (8,5.3);  % separator

\draw(9,4) -- (9,5);
\draw(10,4) -- (10,5);
\draw(11,4) -- (11,5);
\draw(12, 4) -- (12,5.0);
\draw(13, 4) -- (13,5.0);

% following new: separators between segments:

%\draw[line width=0.35mm] (15,3) -- (15,5.7);  % separator
%\draw[line width=0.35mm,white] (15,3) -- (15,4);  % separator
\draw[line width=0.35mm] (15,2.3) -- (15,5.7);  % separator

\draw[line width = 0.35mm](14,3) -- (14,5.7);  % separator
\draw[line width=0.35mm,white] (14,3) -- (14,4);  % separator
\draw[line width = 0.35mm](14,2.3) -- (14,5.7);  % separator

\draw[line width=0.35mm](12,2.3) -- (12,5.7);  % separator

\draw[line width=0.35mm](8,2.3) -- (8,5.7);  % separator

\draw[line width=0.35mm](0,2.3) -- (0,5.7);  % separator

%-------------------------------------------------------- Values in the cell

\node at (8.5,4.5) {\large $83$};
\node at (9.5,4.5) {\large $67$};
\node at (10.5,4.5) {\large $59$};
\node at (11.5,4.5) {\large $21$};
\node at (12.5,4.5) {\large $76$};
\node at (13.5,4.5) {\large $33$};
\node at (14.5,4.5) {\large $45$};
\node at (14.5,4.5) {\large $45$};
\node at (14.5,3.5){\large $52$};
\draw(14.2,3.5)-- (14.8,3.5);
\node at (14.5,3.5){\large $52$};
\draw(14.2,4.5)-- (14.8,4.5);
\node at (13.5,3.5) {\large $45$};
\node at (12.5,3.5){\large $52$};

\node at (8.5,3.5) {\large $76$};
\node at (9.5,3.5){\large $52$};
\node at (10.5,3.5) {\large $45$};
\node at (11.5,3.5){\large $33$};
\draw(12.2,4.5)-- (13.8,4.5);
\draw(12.2,3.5)-- (13.8,3.5);

%----------------------------------------------- Notational Labels:
\node at(-0.8, 5.4) {\it rank:};

\node at(-0.8, 2.55) {\it index:};

\node at(-0.8, 4.5) {\scriptsize WHITE};
\node at(-0.8, 3.5) {\scriptsize BLACK};
\node at (4, 5.5) {\large \bf 3};
\node at (10, 5.5) {\large \bf 2};
\node at (13, 5.5) {\large \bf 1};
\node at (14.5, 5.5) {\large \bf 0};

% -------------------  ending cell cross:
\draw(15,4) -- (16,5);
\draw(15,5) -- (16,4);
\draw(15,3) -- (16,4);
\draw(15,4) -- (16,3);

% ---------------------these are the index values
\node at (0.5,2.5) {$\large15$};
\node at (1.5,2.5) {\large $14$};
\node at (2.5,2.5) {\large $13$};
\node at (3.5,2.5) {\large $12$};

\node at (4.5,2.5) {\large $11$};
\node at (5.5,2.5) {\large $10$};
\node at (6.5,2.5) {\large $9$};
\node at (7.5,2.5) {\large $8$};

\node at (8.5,2.5) {\large $7$};
\node at (9.5,2.5) {\large $6$};
\node at (10.5,2.5) {\large $5$};
\node at (11.5,2.5) {\large $4$};

\node at (12.5,2.5) {\large $3$};
\node at (13.5,2.5) {\large $2$};
\node at (14.5,2.5) {\large $1$};
\node at (15.5,2.5) {\large $0$};

% Cell separators for lower rectangular, (in fact can be merged with the upper one)
\draw(9,3) -- (9,4);
\draw(10,3) -- (10,4);
\draw(11,3) -- (11,4);
\draw(12, 3) -- (12,4);
\draw(13, 3) -- (13,4);

% patent filing labels
%\draw[line width=1](10,2.4) -- (10.7,1.8) -- (10.9, 1.8);
%\node at (11.4, 1.85) {903};

\node at(8, 1.3) {\Large ({\bf c}) Merge of $Seg^{w} (1) $ and $Seg^b (1)$   $\Rightarrow$ $Seg^b (2) $};
%\node at(8, 1.5) {Merging white and black segments of rank 0 into black segment of rank 1};

\end{tikzpicture}

&

%\end{figure}  *******
%end{document}

\begin{tikzpicture}[scale=0.9]

% this the upper rectangle and lower ractangular

\draw (0,4) rectangle (16,5);  % White 
\draw [fill=gray] (8,3) rectangle (16,4); % Black

\draw[line width=0.4mm, black] (0,4) -- (16,4);  % separator

% these are vertical lines separating the bits:

\draw(1,4) -- (1,5);
\draw(2,4) -- (2,5);
\draw(3,4) -- (3,5);
\draw(4,4) -- (4,5);

%\draw(4, 3.8) -- (4,5.3);% separator between bit areas

\draw(5,4) -- (5,5);
\draw(6,4) -- (6,5);
\draw(7,4) -- (7,5);
\draw (8,4) -- (8,5);

%\draw(8,3.8) -- (8,5.3);  % separator

\draw(9,4) -- (9,5);
\draw(10,4) -- (10,5);
\draw(11,4) -- (11,5);
\draw(12, 4) -- (12,5.0);
\draw(13, 4) -- (13,5.0);

% following new: separators between segments:

%\draw[line width=0.35mm] (15,3) -- (15,5.7);  % separator
%\draw[line width=0.35mm,white] (15,3) -- (15,4);  % separator
\draw[line width=0.35mm] (15,2.3) -- (15,5.7);  % separator

\draw[line width = 0.35mm](14,3) -- (14,5.7);  % separator
\draw[line width=0.35mm,white] (14,3) -- (14,4);  % separator
\draw[line width = 0.35mm](14,2.3) -- (14,5.7);  % separator

\draw[line width=0.35mm](12,2.3) -- (12,5.7);  % separator

\draw[line width=0.35mm](8,2.3) -- (8,5.7);  % separator

\draw[line width=0.35mm](0,2.3) -- (0,5.7);  % separator

%-------------------------------------------------------- Values in the cell

\node at (8.5,4.5) {\large $83$};
\node at (9.5,4.5) {\large $67$};
\node at (10.5,4.5) {\large $59$};
\node at (11.5,4.5) {\large $21$};
\node at (12.5,4.5) {\large $76$};
\node at (13.5,4.5) {\large $33$};
\node at (14.5,4.5) {\large $45$};
\node at (14.5,4.5) {\large $45$};
\node at (14.5,3.5){\large $54$};
\draw(14.2,3.5)-- (14.8,3.5);
\node at (14.5,3.5){\large $52$};
\draw(14.2,4.5)-- (14.8,4.5);
\node at (13.5,3.5) {\large $45$};
\node at (12.5,3.5){\large $ 52$};

\node at (8.5,3.5) {\large $76$};
\node at (9.5,3.5){\large $52$};
\node at (10.5,3.5) {\large $45$};
\node at (11.5,3.5){\large $33$};
\draw(12.2,4.5)-- (13.8,4.5);
\draw(12.2,3.5)-- (13.8,3.5);

\node at (0.5,4.5) {\large $83$};
\node at (1.5,4.5) {\large $76$};
\node at (2.5,4.5) {\large $67$};
\node at (3.5,4.5) {\large $59$};

\node at (4.5,4.5) {\large $52$};
\node at (5.5,4.5) {\large $45$};
\node at (6.5,4.5) {\large $33$};
\node at (7.5,4.5) {\large $21$};

\draw(8.2,4.5)-- (11.8,4.5);
\draw(8.2,3.5)-- (11.8,3.5);

%----------------------------------------------- Notational Labels:
\node at(-0.8, 5.4) {\it rank:};

\node at(-0.8, 2.55) {\it index:};

\node at(-0.8, 4.5) {\scriptsize WHITE};
\node at(-0.8, 3.5) {\scriptsize BLACK};
\node at (4, 5.5) {\large \bf 3};
\node at (10, 5.5) {\large \bf 2};
\node at (13, 5.5) {\large \bf 1};
\node at (14.5, 5.5) {\large \bf 0};

% -------------------  ending cell cross:
\draw(15,4) -- (16,5);
\draw(15,5) -- (16,4);
\draw(15,3) -- (16,4);
\draw(15,4) -- (16,3);

% ---------------------these are the index values
\node at (0.5,2.5) {\large $15$};
\node at (1.5,2.5) {\large $14$};
\node at (2.5,2.5) {\large $13$};
\node at (3.5,2.5) {\large $12$};

\node at (4.5,2.5) {\large $11$};
\node at (5.5,2.5) {\large $10$};
\node at (6.5,2.5) {\large $9$};
\node at (7.5,2.5) {\large $8$};

\node at (8.5,2.5) {\large $7$};
\node at (9.5,2.5) {\large $6$};
\node at (10.5,2.5) {\large $5$};
\node at (11.5,2.5) {\large $4$};

\node at (12.5,2.5) {\large $3$};
\node at (13.5,2.5) {\large $2$};
\node at (14.5,2.5) {\large $1$};
\node at (15.5,2.5) {\large $0$};

% Cell separators for lower rectangular, (in fact can be merged with the upper one)
\draw(9,3) -- (9,4);
\draw(10,3) -- (10,4);
\draw(11,3) -- (11,4);
\draw(12, 3) -- (12,4);
\draw(13, 3) -- (13,4);

%\draw[line width=1](4,3.7) -- (4.7,3.2) -- (4.9, 3.2);

\node at (4,3.6) {$\underbrace{\hspace{16em}}^{\ }$};

%\node at (4, 3.15) {904};

\node at(8, 1.2) {\Large ({\bf d}) Merge of  $Seg^{w} (2) $ and $Seg^b (2)$   $\Rightarrow$ $Seg^w (3) $};
\end{tikzpicture}

%

%\end{tikzpicture}

\end{tabular}

%\caption{The insertion of an element invoked recursion process}
%\label{fig:insert}
%\end{figure}

%% file: BWA_delete.tex
%\begin{figure}

%\resizebox{0.5\textwidth}{!}
%\begin{tikzpicture}[scale=0.9]

\begin{tabular}{cc}
\begin{tikzpicture}[scale=0.9]

% this the upper rectangle and lower ractangular

\draw (0,4) rectangle (16,5);  % White 
\draw [fill=gray] (8,3) rectangle (16,4); % Black

\draw[line width=0.4mm, black] (0,4) -- (16,4);  % separator

% these are vertical lines separating the bits:

\draw(1,4) -- (1,5);
\draw(2,4) -- (2,5);
\draw(3,4) -- (3,5);
\draw(4,4) -- (4,5);

%\draw(4, 3.8) -- (4,5.3);% separator between bit areas

\draw(5,4) -- (5,5);
\draw(6,4) -- (6,5);
\draw(7,4) -- (7,5);
\draw (8,4) -- (8,5);

%\draw(8,3.8) -- (8,5.3);  % separator

\draw(9,4) -- (9,5);
\draw(10,4) -- (10,5);
\draw(11,4) -- (11,5);
\draw(12, 4) -- (12,5.0);
\draw(13, 4) -- (13,5.0);

% following new: separators between segments:

%\draw[line width=0.35mm] (15,3) -- (15,5.7);  % separator
%\draw[line width=0.35mm,white] (15,3) -- (15,4);  % separator
\draw[line width=0.35mm] (15,2.3) -- (15,5.7);  % separator

\draw[line width = 0.35mm](14,3) -- (14,5.7);  % separator
\draw[line width=0.35mm,white] (14,3) -- (14,4);  % separator
\draw[line width = 0.35mm](14,2.3) -- (14,5.7);  % separator

\draw[line width=0.35mm](12,2.3) -- (12,5.7);  % separator

\draw[line width=0.35mm](8,2.3) -- (8,5.7);  % separator

\draw[line width=0.35mm](0,2.3) -- (0,5.7);  % separator

%-------------------------------------------------------- Values in the cell

\node at (0.5,4.5) {\large$83$};
\node at (1.5,4.5) {\large$\phi$};
\node at (2.5,4.5) {\large$67$};
\node at (3.5,4.5) {\large$59$};
\node at (4.5,4.5) {\large$52$};
\node at (5.5,4.5) {\large$\phi$};
\node at (6.5,4.5) {\large$\phi$};
\node at (7.5,4.5) {\large$6$};

\node at (8.5,4.5) {\large$91$};
\node at (9.5,4.5) {\large $\phi$};
\node at (10.5,4.5) {\large $77$};
\node at (11.5,4.5) {\large $21$};

\node at (12.5,4.5) {\large $82$};
\node at (13.5,4.5) {\large $45$};
%\node at (14.5,4.5) {\large $45$};
%\draw(14,4.5)-- (15,4.5);
%\node at (14.5,3.5){\large $82$};
%\draw(14,3.5)-- (15,3.5);

%----------------------------------------------- Notational Labels:
\node at(-0.8, 5.4) {\it rank:};

\node at(-0.8, 2.55) {\it index:};

\node at(-0.8, 4.5) {\scriptsize WHITE};
\node at(-0.8, 3.5) {\scriptsize BLACK};
\node at (4, 5.5) {\large \bf 3};
\node at (10, 5.5) {\large \bf 2};
\node at (13, 5.5) {\large \bf 1};
\node at (14.5, 5.5) {\large \bf 0};

% -------------------  ending cell cross:
\draw(15,4) -- (16,5);
\draw(15,5) -- (16,4);
\draw(15,3) -- (16,4);
\draw(15,4) -- (16,3);

% ---------------------these are the index values
\node at (0.5,2.5) {\large$15$};
\node at (1.5,2.5) {\large $14$};
\node at (2.5,2.5) {\large $13$};
\node at (3.5,2.5) {\large $12$};

\node at (4.5,2.5) {\large $11$};
\node at (5.5,2.5) {\large $10$};
\node at (6.5,2.5) {\large $9$};
\node at (7.5,2.5) {\large $8$};

\node at (8.5,2.5) {\large $7$};
\node at (9.5,2.5) {\large $6$};
\node at (10.5,2.5) {\large $5$};
\node at (11.5,2.5) {\large $4$};

\node at (12.5,2.5) {\large $3$};
\node at (13.5,2.5) {\large $2$};
\node at (14.5,2.5) {\large $1$};
\node at (15.5,2.5) {\large $0$};

% Cell separators for lower rectangular, (in fact can be merged with the upper one)
\draw(9,3) -- (9,4);
\draw(10,3) -- (10,4);
\draw(11,3) -- (11,4);
\draw(12, 3) -- (12,4);
\draw(13, 3) -- (13,4);

% the following are the labels for the patent filing
%\draw[line width=1](14.5,2.2) -- (15,1.7) -- (15.2, 1.7);
%\node at (15.6, 1.7) {901};

\node at(8, 1.3) {\Large ({\bf a}) A BWA with some VOID ($\phi$) values from deletions };

\end{tikzpicture}

&

\begin{tikzpicture}[scale=0.9]

% this the upper rectangle and lower ractangular

\draw (0,4) rectangle (16,5);  % White 
\draw [fill=gray] (8,3) rectangle (16,4); % Black

\draw[line width=0.4mm, black] (0,4) -- (16,4);  % separator

% these are vertical lines separating the bits:

\draw(1,4) -- (1,5);
\draw(2,4) -- (2,5);
\draw(3,4) -- (3,5);
\draw(4,4) -- (4,5);

%\draw(4, 3.8) -- (4,5.3);% separator between bit areas

\draw(5,4) -- (5,5);
\draw(6,4) -- (6,5);
\draw(7,4) -- (7,5);
\draw (8,4) -- (8,5);

%\draw(8,3.8) -- (8,5.3);  % separator

\draw(9,4) -- (9,5);
\draw(10,4) -- (10,5);
\draw(11,4) -- (11,5);
\draw(12, 4) -- (12,5.0);
\draw(13, 4) -- (13,5.0);

% following new: separators between segments:

%\draw[line width=0.35mm] (15,3) -- (15,5.7);  % separator
%\draw[line width=0.35mm,white] (15,3) -- (15,4);  % separator
\draw[line width=0.35mm] (15,2.3) -- (15,5.7);  % separator

\draw[line width = 0.35mm](14,3) -- (14,5.7);  % separator
\draw[line width=0.35mm,white] (14,3) -- (14,4);  % separator
\draw[line width = 0.35mm](14,2.3) -- (14,5.7);  % separator

\draw[line width=0.35mm](12,2.3) -- (12,5.7);  % separator

\draw[line width=0.35mm](8,2.3) -- (8,5.7);  % separator

\draw[line width=0.35mm](0,2.3) -- (0,5.7);  % separator

%-------------------------------------------------------- Values in the cell

\node at (0.5,4.5) {\large$83$};
\node at (1.5,4.5) {\large$\phi$};
\node at (2.5,4.5) {\large$67$};
\node at (3.5,4.5) {\large$\phi$};
\node at (4.5,4.5) {\large$52$};
\node at (5.5,4.5) {\large$\phi$};
\node at (6.5,4.5) {\large$\phi$};
\node at (7.5,4.5) {\large$6$};

\node at (8.5,4.5) {\large$91$};
\node at (9.5,4.5) {\large $\phi$};
\node at (10.5,4.5) {\large $77$};
\node at (11.5,4.5) {\large $35$};
\node at (12.5,4.5) {\large $82$};
\node at (13.5,4.5) {\large $45$};

%----------------------------------------------- Notational Labels:
\node at(-0.8, 5.4) {\it  rank:};

\node at(-0.8, 2.55) {\it index:};

\node at(-0.8, 4.5) {\scriptsize WHITE};
\node at(-0.8, 3.5) { \scriptsize BLACK};
\node at (4, 5.5) {\large \bf 3};
\node at (10, 5.5) {\large \bf 2};
\node at (13, 5.5) {\large \bf 1};
\node at (14.5, 5.5) {\large \bf 0};

% -------------------  ending cell cross:
\draw(15,4) -- (16,5);
\draw(15,5) -- (16,4);
\draw(15,3) -- (16,4);
\draw(15,4) -- (16,3);

% ---------------------these are the index values
\node at (0.5,2.5) {\large$15$};
\node at (1.5,2.5) {\large $14$};
\node at (2.5,2.5) {\large $13$};
\node at (3.5,2.5) {\large $12$};

\node at (4.5,2.5) {\large $11$};
\node at (5.5,2.5) {\large $10$};
\node at (6.5,2.5) {\large $9$};
\node at (7.5,2.5) {\large $8$};

\node at (8.5,2.5) {\large $7$};
\node at (9.5,2.5) {\large $6$};
\node at (10.5,2.5) {\large $5$};
\node at (11.5,2.5) {\large $4$};

\node at (12.5,2.5) {\large $3$};
\node at (13.5,2.5) {\large $2$};
\node at (14.5,2.5) {\large $1$};
\node at (15.5,2.5) {\large $0$};

% Cell separators for lower rectangular, (in fact can be merged with the upper one)
\draw(9,3) -- (9,4);
\draw(10,3) -- (10,4);
\draw(11,3) -- (11,4);
\draw(12, 3) -- (12,4);
\draw(13, 3) -- (13,4);

% patent filing labels
%\draw[line width=1](13,2.2) -- (13.5,1.7) -- (13.7, 1.7);
%\node at (14.2, 1.7) {902};

\node at(8, 1.3) {\Large ({\bf b}) 59 is deleted and replaced by VOID ($\phi$) at W[12]};

\end{tikzpicture}

\\

%----------------------------------------------------------------------------------------------------------------------------
\begin{tikzpicture}[scale=0.9]

% this the upper rectangle and lower ractangular

\draw (0,4) rectangle (16,5);  % White 
\draw [fill=gray] (8,3) rectangle (16,4); % Black

\draw[line width=0.4mm, black] (0,4) -- (16,4);  % separator

% these are vertical lines separating the bits:

\draw(1,4) -- (1,5);
\draw(2,4) -- (2,5);
\draw(3,4) -- (3,5);
\draw(4,4) -- (4,5);

%\draw(4, 3.8) -- (4,5.3);% separator between bit areas

\draw(5,4) -- (5,5);
\draw(6,4) -- (6,5);
\draw(7,4) -- (7,5);
\draw (8,4) -- (8,5);

%\draw(8,3.8) -- (8,5.3);  % separator

\draw(9,4) -- (9,5);
\draw(10,4) -- (10,5);
\draw(11,4) -- (11,5);
\draw(12, 4) -- (12,5.0);
\draw(13, 4) -- (13,5.0);

% following new: separators between segments:

%\draw[line width=0.35mm] (15,3) -- (15,5.7);  % separator
%\draw[line width=0.35mm,white] (15,3) -- (15,4);  % separator
\draw[line width=0.35mm] (15,2.3) -- (15,5.7);  % separator

\draw[line width = 0.35mm](14,3) -- (14,5.7);  % separator
\draw[line width=0.35mm,white] (14,3) -- (14,4);  % separator
\draw[line width = 0.35mm](14,2.3) -- (14,5.7);  % separator

\draw[line width=0.35mm](12,2.3) -- (12,5.7);  % separator

\draw[line width=0.35mm](8,2.3) -- (8,5.7);  % separator

\draw[line width=0.35mm](0,2.3) -- (0,5.7);  % separator

%-------------------------------------------------------- Values in the cell

\node at (8.5,4.5) {\large $91$};
\node at (9.5,4.5) {\large $\phi$};
\node at (10.5,4.5) {\large $77$};
\node at (11.5,4.5) {\large $21$};
\node at (12.5,4.5) {\large $82$};
\node at (13.5,4.5) {\large $45$};
%\node at (14.5,4.5) {\large $45$};
%\node at (14.5,4.5) {\large $45$};
%\node at (14.5,3.5){\large $91$};
%\draw(14.2,3.5)-- (14.8,3.5);
%\node at (14.5,3.5){\large $77$};
%\draw(14.2,4.5)-- (14.8,4.5);
\node at (13.5,3.5) {\large $45$};
\node at (12.5,3.5){\large $52$};

\node at (8.5,3.5) {\large $83$};
\node at (9.5,3.5){\large $67$};
\node at (10.5,3.5) {\large $52$};
\node at (11.5,3.5){\large $6$};
%\draw(12.2,4.5)-- (13.8,4.5);
%\draw(12.2,3.5)-- (13.8,3.5);

%----------------------------------------------- Notational Labels:
\node at(-0.8, 5.4) {\it rank:};

\node at(-0.8, 2.55) {\it index:};

\node at(-0.8, 4.5) {\scriptsize WHITE};
\node at(-0.8, 3.5) {\scriptsize BLACK};
\node at (4, 5.5) {\large \bf 3};
\node at (10, 5.5) {\large \bf 2};
\node at (13, 5.5) {\large \bf 1};
\node at (14.5, 5.5) {\large \bf 0};

% -------------------  ending cell cross:
\draw(15,4) -- (16,5);
\draw(15,5) -- (16,4);
\draw(15,3) -- (16,4);
\draw(15,4) -- (16,3);

% ---------------------these are the index values
\node at (0.5,2.5) {$\large15$};
\node at (1.5,2.5) {\large $14$};
\node at (2.5,2.5) {\large $13$};
\node at (3.5,2.5) {\large $12$};

\node at (4.5,2.5) {\large $11$};
\node at (5.5,2.5) {\large $10$};
\node at (6.5,2.5) {\large $9$};
\node at (7.5,2.5) {\large $8$};

\node at (8.5,2.5) {\large $7$};
\node at (9.5,2.5) {\large $6$};
\node at (10.5,2.5) {\large $5$};
\node at (11.5,2.5) {\large $4$};

\node at (12.5,2.5) {\large $3$};
\node at (13.5,2.5) {\large $2$};
\node at (14.5,2.5) {\large $1$};
\node at (15.5,2.5) {\large $0$};

% Cell separators for lower rectangular, (in fact can be merged with the upper one)
\draw(9,3) -- (9,4);
\draw(10,3) -- (10,4);
\draw(11,3) -- (11,4);
\draw(12, 3) -- (12,4);
\draw(13, 3) -- (13,4);

% patent filing labels
%\draw[line width=1](10,2.4) -- (10.7,1.8) -- (10.9, 1.8);
%\node at (11.4, 1.85) {903};

\node at(8, 1.3) {\Large ({\bf c}) Seg(3) was demoted for its occupancy rate reached $50\%$};
%\node at(8, 1.5) {Merging white and black segments of rank 0 into black segment of rank 1};

\end{tikzpicture}

&

%\end{figure}  *******
%end{document}

\begin{tikzpicture}[scale=0.9]

% this the upper rectangle and lower ractangular

\draw (0,4) rectangle (16,5);  % White 
\draw [fill=gray] (8,3) rectangle (16,4); % Black

\draw[line width=0.4mm, black] (0,4) -- (16,4);  % separator

% these are vertical lines separating the bits:

\draw(1,4) -- (1,5);
\draw(2,4) -- (2,5);
\draw(3,4) -- (3,5);
\draw(4,4) -- (4,5);

%\draw(4, 3.8) -- (4,5.3);% separator between bit areas

\draw(5,4) -- (5,5);
\draw(6,4) -- (6,5);
\draw(7,4) -- (7,5);
\draw (8,4) -- (8,5);

%\draw(8,3.8) -- (8,5.3);  % separator

\draw(9,4) -- (9,5);
\draw(10,4) -- (10,5);
\draw(11,4) -- (11,5);
\draw(12, 4) -- (12,5.0);
\draw(13, 4) -- (13,5.0);

% following new: separators between segments:

%\draw[line width=0.35mm] (15,3) -- (15,5.7);  % separator
%\draw[line width=0.35mm,white] (15,3) -- (15,4);  % separator
\draw[line width=0.35mm] (15,2.3) -- (15,5.7);  % separator

\draw[line width = 0.35mm](14,3) -- (14,5.7);  % separator
\draw[line width=0.35mm,white] (14,3) -- (14,4);  % separator
\draw[line width = 0.35mm](14,2.3) -- (14,5.7);  % separator

\draw[line width=0.35mm](12,2.3) -- (12,5.7);  % separator

\draw[line width=0.35mm](8,2.3) -- (8,5.7);  % separator

\draw[line width=0.35mm](0,2.3) -- (0,5.7);  % separator

%-------------------------------------------------------- Values in the cell

\node at (0.5,4.5) {\large $\phi$};
\node at (1.5,4.5) {\large $91$};
\node at (2.5,4.5) {\large $83$};
\node at (3.5,4.5) {\large $77$};

\node at (4.5,4.5) {\large $67$};
\node at (5.5,4.5) {\large $52$};
\node at (6.5,4.5) {\large $21$};
\node at (7.5,4.5) {\large $6$};

\node at (12.5,4.5) {\large $82$};
\node at (13.5,4.5) {\large $45$};

%\draw(8.2,4.5)-- (11.8,4.5);
%\draw(8.2,3.5)-- (11.8,3.5);

%----------------------------------------------- Notational Labels:
\node at(-0.8, 5.4) {\it rank:};

\node at(-0.8, 2.55) {\it index:};

\node at(-0.8, 4.5) {\scriptsize WHITE};
\node at(-0.8, 3.5) {\scriptsize BLACK};
\node at (4, 5.5) {\large \bf 3};
\node at (10, 5.5) {\large \bf 2};
\node at (13, 5.5) {\large \bf 1};
\node at (14.5, 5.5) {\large \bf 0};

% -------------------  ending cell cross:
\draw(15,4) -- (16,5);
\draw(15,5) -- (16,4);
\draw(15,3) -- (16,4);
\draw(15,4) -- (16,3);

% ---------------------these are the index values
\node at (0.5,2.5) {\large $15$};
\node at (1.5,2.5) {\large $14$};
\node at (2.5,2.5) {\large $13$};
\node at (3.5,2.5) {\large $12$};

\node at (4.5,2.5) {\large $11$};
\node at (5.5,2.5) {\large $10$};
\node at (6.5,2.5) {\large $9$};
\node at (7.5,2.5) {\large $8$};

\node at (8.5,2.5) {\large $7$};
\node at (9.5,2.5) {\large $6$};
\node at (10.5,2.5) {\large $5$};
\node at (11.5,2.5) {\large $4$};

\node at (12.5,2.5) {\large $3$};
\node at (13.5,2.5) {\large $2$};
\node at (14.5,2.5) {\large $1$};
\node at (15.5,2.5) {\large $0$};

% Cell separators for lower rectangular, (in fact can be merged with the upper one)
\draw(9,3) -- (9,4);
\draw(10,3) -- (10,4);
\draw(11,3) -- (11,4);
\draw(12, 3) -- (12,4);
\draw(13, 3) -- (13,4);

%\draw[line width=1](4,3.7) -- (4.7,3.2) -- (4.9, 3.2);

% the following are labels in patent filing
%\node at (4,3.6) {$\underbrace{\hspace{16em}}^{\ }$};
%\node at (4, 3.15) {904};

\node at(8, 1.2) {\Large ({\bf d}) After a merge at rank 2, the delete is completed};
\end{tikzpicture}

%

%\end{tikzpicture}

\end{tabular}

%\caption{The insertion of an element invoked recursion process}
%\label{fig:insert}
%\end{figure}

%% file: BWATree.tex
%\begin{figure}
%\begin{center}

\tikzset{
  solid node/.style={circle,draw,inner sep=1.2,fill=black},
  hollow node/.style={circle,draw,inner sep=1.2},
}

\begin{tikzpicture}[nodes ={draw, circle},  level distance=1.5cm,
  level 1/.style={sibling distance=4cm},
  level 2/.style={sibling distance=2cm},
  level 3/.style={sibling distance = 1cm},
  left/.style={circle,fill=black!20,font=\ttfamily},
  seg/.style={ font=\ttfamily},
  seg2/.style={rectangle, font=\ttfamily},
  ]

  \node[label=left:{\scriptsize rank 3}]{\tiny 15}
  %\node[label=left:{\small test}]{15}
      child {node [label=left:{\scriptsize rank 2}]{\scriptsize 7}
                child {node [label=left:{\scriptsize rank 1}]{\scriptsize 3}
                         child {node [label=left:{\scriptsize rank 0}]{\scriptsize 1}}
                         child {node [left] {\scriptsize 2}}
                 }
                child {node [left] {\scriptsize 6}
                         child {node {\scriptsize 4}}
                         child {node [left] {\scriptsize 5}}
                }
       }
      child {node[left] {\tiny 14}
             child {node {\tiny 10}
                       child {node {\scriptsize 8}}
                       child {node[left] {\scriptsize  9}}
              }
             child {node[left] {\tiny 13}
                       child {node {\tiny 11}}
                       child {node [left] {\tiny 12}}
              }
     }; 
\end{tikzpicture}
%\end{center}
%\vspace{-6mm}
%\caption{The operational ordering in Black-White Array}
%\end{figure}

%% file: MergesortTree.tex
%\begin{figure}
\begin{center}
\begin{tikzpicture}[nodes ={draw, circle},  level distance=1.5cm,
  level 1/.style={sibling distance=4cm},
  level 2/.style={sibling distance=2cm},
  level 3/.style={sibling distance = 1cm},
  left/.style={circle,fill=black!20,font=\ttfamily},
  seg/.style={ font=\ttfamily},
  seg2/.style={rectangle, font=\ttfamily},
  ]

  \node[left][label=left:{\scriptsize level 3}]{\tiny 15}
  %\node[label=left:{\small test}]{15}
      child {node[label=left:{\scriptsize level 2}] {\tiny 13}
                child {node[left][label=left:{\scriptsize level  1}] {\scriptsize 9}
                         child {node[label=left:{\scriptsize level 0}] {\scriptsize 1}}
                         child {node {\scriptsize 2}}
                 }
                child {node[left] {\tiny 10}
                         child {node {\scriptsize 3}}
                         child {node  {\scriptsize 4}}
                }
       }
      child {node {\tiny 14}
             child {node[left] {\tiny 11}
                       child {node {\scriptsize 5}}
                       child {node {\scriptsize 6}}
              }
             child {node[left] {\tiny 12}
                       child {node {\tiny 7}}
                       child {node  {\tiny 8}}
              }
     }; 
\end{tikzpicture}
\end{center}
%\vspace{-6mm}
%\caption{The operational ordering in Merge Sort}
%\end{figure}